\documentclass[12pt]{article}

\usepackage[latin1]{inputenc} 
\usepackage{amssymb} 
\usepackage{amsmath}
\usepackage{amsthm}
\usepackage{mathtools}
\usepackage{txfonts} 
\usepackage{latexsym}
\usepackage{subfigure}
\usepackage{graphicx}
\usepackage{bm}
\usepackage{bbm}
\usepackage{overpic}
\usepackage[normalem]{ulem}
\usepackage{url}
\usepackage{soul}
\usepackage{exscale}
\usepackage{amsfonts}
\usepackage[usenames,dvipsnames]{xcolor} % load color package
\usepackage{hyperref}
\usepackage{verbatim}
\usepackage[usenames,dvipsnames]{xcolor}
\usepackage{comment}
\usepackage{epsfig}

\newtheorem{theorem}{Theorem}[section]
\newtheorem{lemma}[theorem]{Lemma}
\newtheorem{prop}[theorem]{Proposition}
\newtheorem{cor}[theorem]{Corollary}

\newtheorem{definition}[theorem]{Definition}
\newtheorem{remark}[theorem]{Remark}

\numberwithin{table}{section}
\numberwithin{equation}{section}
\numberwithin{figure}{section}

\numberwithin{equation}{section}
\numberwithin{figure}{section}

\newcommand{\E}{\mathbb{E}}
\newcommand{\R}{\mathbb{R}}
 
\newcommand{\BS}{\rm BS}

\newcommand{\p}{\partial}

\newcommand{\beas}{\begin{eqnarray*}}
\newcommand{\eeas}{\end{eqnarray*}}
\newcommand{\bal}{\begin{align}}
\newcommand{\eal}{\end{align}}
\newcommand{\bas}{\begin{align*}}
\newcommand{\eas}{\end{align*}}
\newcommand{\bea}{\begin{eqnarray}}
\newcommand{\eea}{\end{eqnarray}}
\newcommand{\tmop}{\end{eqnarray}}
\newcommand{\ben}{\begin{enumerate}}
\newcommand{\een}{\end{enumerate}}

\newcommand{\ui}{\mathrm{i}}

\newcommand{\dm}{\diamond}

\newcommand{\cR}{\mathcal{R}}

\newcommand{\cS}{\mathcal{S}}

\newcommand{\cF}{\mathcal{F}}

\newcommand{\mF}{\mathbb{F}}
\newcommand{\mT}{\mathbb{T}}

\newcommand{\tmF}{\tilde {\mathbb{F}}}
\newcommand{\mP}{\mathbb{P}}
\newcommand{\mQ}{\mathbb{Q}}
\newcommand{\IF}{It\^o's Formula}
\newcommand{\cO}{\mathcal{O}}

\newcommand{\bi}{\begin{itemize}}
\newcommand{\ei}{\end{itemize}}
\newcommand{\beq}{\begin{equation}}
\newcommand{\eeq}{\end{equation}}
\newcommand{\bv}{\begin{verbatim}}
\newcommand{\ev}{\end{verbatim}}

\newcommand{\eef}[1]{\ensuremath{\mathbb{E}\left[\left.{#1}\right|\cF_t\right]}}

\newcommand{\eetm}[2]{\mathbb{E}_t^{#2}\left[#1\right]}
\newcommand{\eet}[1]{\mathbb{E}_{t}\left[#1\right]}
\newcommand{\angl}[1]{\left\langle{#1}\right\rangle}

\usepackage{tikz}
\newcommand{\tkz}{\tikzexternaldisable}
\usepackage{mhequ} 
\usetikzlibrary{snakes}
\usetikzlibrary{decorations}
\usetikzlibrary{positioning}
\usetikzlibrary{shapes}
\usetikzlibrary{external}
\usetikzlibrary{snakes}
\usetikzlibrary{decorations}
\usetikzlibrary{positioning}
\usetikzlibrary{shapes}
\usetikzlibrary{external}
\makeatletter
\pgfdeclareshape{crosscircle}
{
  \inheritsavedanchors[from=circle] % this is nearly a circle
  \inheritanchorborder[from=circle]
  \inheritanchor[from=circle]{north}
  \inheritanchor[from=circle]{north west}
  \inheritanchor[from=circle]{north east}
  \inheritanchor[from=circle]{center}
  \inheritanchor[from=circle]{west}
  \inheritanchor[from=circle]{east}
  \inheritanchor[from=circle]{mid}
  \inheritanchor[from=circle]{mid west}
  \inheritanchor[from=circle]{mid east}
  \inheritanchor[from=circle]{base}
  \inheritanchor[from=circle]{base west}
  \inheritanchor[from=circle]{base east}
  \inheritanchor[from=circle]{south}
  \inheritanchor[from=circle]{south west}
  \inheritanchor[from=circle]{south east}
  \inheritbackgroundpath[from=circle]
  \foregroundpath{
    \centerpoint%
    \pgf@xc=\pgf@x%
    \pgf@yc=\pgf@y%
    \pgfutil@tempdima=\radius%
    \pgfmathsetlength{\pgf@xb}{\pgfkeysvalueof{/pgf/outer xsep}}%  
    \pgfmathsetlength{\pgf@yb}{\pgfkeysvalueof{/pgf/outer ysep}}%  
    \ifdim\pgf@xb<\pgf@yb%
      \advance\pgfutil@tempdima by-\pgf@yb%
    \else%
      \advance\pgfutil@tempdima by-\pgf@xb%
    \fi%
    \pgfpathmoveto{\pgfpointadd{\pgfqpoint{\pgf@xc}{\pgf@yc}}{\pgfqpoint{-0.707107\pgfutil@tempdima}{-0.707107\pgfutil@tempdima}}}
    \pgfpathlineto{\pgfpointadd{\pgfqpoint{\pgf@xc}{\pgf@yc}}{\pgfqpoint{0.707107\pgfutil@tempdima}{0.707107\pgfutil@tempdima}}}
    \pgfpathmoveto{\pgfpointadd{\pgfqpoint{\pgf@xc}{\pgf@yc}}{\pgfqpoint{-0.707107\pgfutil@tempdima}{0.707107\pgfutil@tempdima}}}
    \pgfpathlineto{\pgfpointadd{\pgfqpoint{\pgf@xc}{\pgf@yc}}{\pgfqpoint{0.707107\pgfutil@tempdima}{-0.707107\pgfutil@tempdima}}}
  }
}
\makeatother

\def\S{\tikz[baseline=-2.8,scale=0.15]{\node[T1] {};}} % BrickRed colored leaves for T1 
\def\U{\tikz[baseline=-2.8,scale=0.15]{\node[T2] {};}} % BurntOrange colored leaves for T2
\def\X{\tikz[baseline=-2.8,scale=0.15]{\node[X] {};}} % Gray colored leaves for X
\def\M{\tikz[baseline=-2.8,scale=0.15]{\node[M] {};}} % Brown colored leaves for M = X \dm X
\def\XXd{\tikz[baseline=-1,scale=0.15]{\draw (-1,1) node[X] {} -- (0,0) node[not] {} -- (1,1) node[X] {};}} % new for Jim
\def\MMd{\tikz[baseline=-1,scale=0.15]{\draw (-1,1) node[M] {} -- (0,0) node[not] {} -- (1,1) node[M] {};}} % new for Jim
\def\T1T2d{\tikz[baseline=-1,scale=0.15]{\draw (-1,1) node[T1] {} -- (0,0) node[not] {} -- (1,1) node[T2] {};}} % 
\def\MXd{\tikz[baseline=-1,scale=0.15]{\draw (-1,1) node[M] {} -- (0,0) node[not] {} -- (1,1) node[X] {};}}

\def\MXdXd{\tikz[baseline=-1,scale=0.15]{
\draw (0,0) node[not] {} -- (-1,1) node[not] {}
-- (-2,2) node[M]{} ;
\draw (0,0) -- (1,1) node[X] {};
\draw (-1,1) -- (0,2) node[X] {};
}}

\def\XMdMd{\tikz[baseline=-1,scale=0.15]{
\draw (0,0) node[not] {} -- (-1,1) node[not] {}
-- (-2,2) node[X]{} ;
\draw (0,0) -- (1,1) node[M] {};
\draw (-1,1) -- (0,2) node[M] {};
}}

\def\MMdXd{\tikz[baseline=-1,scale=0.15]{
\draw (0,0) node[not] {} -- (-1,1) node[not] {}
-- (-2,2) node[M]{} ;
\draw (0,0) -- (1,1) node[X] {};
\draw (-1,1) -- (0,2) node[M] {};
}}

\def\MXdXdXd{\tikz[baseline=-1,scale=0.15]{
\draw (0,0) node[not] {} -- (-1,1) node[not] {}
-- (-2,2) node[not]{}  -- (-3,3) node[M]{};
\draw (0,0) -- (1,1) node[X] {};
\draw (-1,1) -- (0,2) node[X] {};
\draw (-2,2) -- (-1,3) node[X] {};
}}

\colorlet{symbols}{blue!90!black}
\colorlet{testcolor}{green!60!black}
\colorlet{connection}{red!30!black}

\tikzset{
root/.style={circle,fill=black!50,inner sep=0pt, minimum size=3mm},
        dot/.style={circle,fill=black,inner sep=0pt, minimum size=1.5mm},
        T1/.style={very thin,circle,fill=BrickRed!100,draw=black,inner sep=0pt,minimum size=1.2mm},         		T2/.style={very thin,circle,fill=BurntOrange!100,draw=black,inner sep=0pt,minimum size=1.2mm},         	A/.style={very thin,circle,fill=BlueGreen!100,draw=black,inner sep=0pt,minimum size=1.2mm},         		X/.style={very thin,circle,fill=Gray!80,draw=black,inner sep=0pt,minimum size=1.2mm},         			Z/.style={very thin,circle,fill=RubineRed!100,draw=black,inner sep=0pt,minimum size=1.2mm},         		M/.style={very thin,circle,fill=YellowOrange!100,draw=black,inner sep=0pt,minimum size=1.2mm},         %xix/.style={crosscircle,fill=blue!10,draw=black,inner sep=0pt,minimum size=1.2mm},
	not/.style={thin,circle,fill=symbols,draw=connection,fill=connection,inner sep=0pt,minimum size=0.35mm},
	>=stealth,
        }
\begin{document}

\title{Computing the SSR}

\author{Peter K. Friz, TU and WIAS Berlin,\\ {\tt friz@math.tu-berlin.edu}\\
~\\
Jim Gatheral, Baruch College, CUNY,\\ {\tt jim.gatheral@baruch.cuny.edu}}

%\date{This Version: June 18, 2024}

\maketitle\thispagestyle{empty}

\begin{abstract}
The skew-stickiness-ratio (SSR), examined in detail by Bergomi in his book, is critically important to options traders, especially market makers.  We present a  model-free expression for the  SSR in terms of the characteristic function.  In the diffusion setting, it is well-known that the short-term limit of the SSR is 2; a corollary of our results is that this limit is $H+3/2$ where $H$ is the Hurst exponent of the volatility process.  The general formula for the SSR simplifies and becomes particularly tractable in the affine forward variance case. We explain the qualitative behavior of the SSR with respect to the shape of the  forward variance curve, and thus also path-dependence of the SSR.

\end{abstract}

%\tableofcontents

\section{The skew-stickiness-ratio (SSR)}\label{sec:intro}

Consider a European call option expiring at $T$ with strike $K$, valued as of time $t$.  
Denote the Black-Scholes implied volatility of this option by $\sigma_{\BS,t}(k,T)$ where $k = \log K/S$ is the log-strike, $S=S_t$ the spot. More explicitly, 
$\sigma_{\BS}=\sigma_{\BS,t}(k,T)$ is the unique solution of %the equation
\[ \eet{\left( \frac{S_T}{S_t} - e^k \right)^+ }  
   = \mathrm{BS} (k, (T -t) \sigma_{\BS}^2), \]
where the Black-Scholes call price function is given by
\[ \mathrm{BS} (k, w) = \left\{\begin{array}{ll}
     \Phi \left( - \frac{k}{\sqrt{w}} + \frac{\sqrt{w}}{2} \right) - e^k \Phi
     \left( - \frac{k}{\sqrt{w}} - \frac{\sqrt{w}}{2} \right) & \text{if } w >
     0\\
     (1 - e^k)^+ & \text{if } w = 0,
   \end{array}\right. \]
and $\Phi$ denotes the standard normal distribution function. Equivalently,
$$
    C_t (K,T) = 
   %\mathbb{E} \left[ \left. \left( {S_T} - K \right)^+ \right\rvert  \hspace{0.17em} \mathcal{F}_t \right] 
   \eet{ \left( {S_T} - K \right)^+}
    = C_{\BS,t} (K,T;\sigma_{\BS})%  = S_t \mathrm{BS} (k, (T -t) \sigma_{\BS}^2),
$$
with $\sigma_{\BS} = \sigma_{\BS,t}(\log (K/S_t),T)$.
To hedge options using the Black-Scholes formula, market makers need to hedge two effects.  First, the explicit spot effect 
$$\frac{\partial C_{\BS,t}}{\partial S_t}\,\delta S_t,$$
and secondly, the change in implied volatility conditional on a change in the spot 
$$
\frac{\partial C_{\BS,t}}{\partial \sigma_{\BS}}\,\E \left[{\delta \sigma_{\BS,t} | \delta S_t}\right].
$$
(The left-hand factor here is just that Black-Scholes vega, evaluated at $\sigma_{\BS}=\sigma_{\BS,t}(k,T)$.)

%Making the observation date $t$ explicit, 
ATM implied volatilities $\sigma_t(T)=\sigma_{\text{BS},t}(0,T)$ and stock prices are both observable.  So market makers can estimate the second component for at-the-money (ATM) options using a simple regression:
$$
\delta \sigma_t(T) =  \beta_t(T)\,\frac{\delta S_t}{S_t} + \text{noise}=: \beta_t(T)\,\delta X_t + \text{noise}.
$$
For a given time to expiration $T$, we define the ATM volatility skew
$$
\cS_t(T) = \left. \frac{\partial}{\partial k} \sigma_{\BS,t}(k,T) \right|_{k=0}.
$$
Obviously, $\cS_t(T)$ is also a market observable.  
Bergomi \cite{bergomi2009smile} calls 
\begin{equation} \label{eq:SSR}
\mathcal{R}_t(T) = \frac{\beta_t(T)}{\cS_t(T)}.
\end{equation}
the {\em skew-stickiness ratio} or {\em SSR}.  

This paper is organized as follows. In Section \ref{sec:intro}, we begin by introducing the concept of the skew-stickiness-ratio (SSR). 
and discuss its relevance in financial modelling.  In Section \ref{sec:stochvol},  the analysis of the SSR is extended to scenarios with stochastic volatility. 
In section \ref{sec:charfn}, the SSR is expressed in terms of the characteristic function, which is the basis of our subsequent analysis.

Section \ref{sec:AFV} specializes to the case of affine forward variance (AFV) models,  with a specific focus on the Classical Heston model. 
Section \ref{sec:forest} details  the computation of \(R_t(T)\) using the forest expansion method, with step-by-step calculations to 
second order and leading order. The examples provided illustrate the application of these methods  to both the Classical Heston model and the Rough Heston model.

Section \ref{sec:pathdep} explores the dependence of \(R_t(T)\) on the forward variance curve and path-dependence, 
addressing how the SSR is influenced by the path characteristics of the underlying model. Finally, in Section \ref{sec:sensitivity},
the sensitivity of the SSR to different model dynamics is discussed.  In Section \ref{sec:conclusion}, we summarize our contributions and make some concluding remarks.

\noindent
{\bf Acknowledgements:} PKF acknowledges funding by the Deutsche Forschungsgemeinschaft (DFG, German Research Foundation) under Germany's Excellence Strategy -- The Berlin Mathematics Research Center MATH+ (EXC-2046/1, project ID: 390685689). Both authors acknowledge seed support for the DFG CRC/TRR 388 ``Rough Analysis, Stochastic'' of which JG is Mercator Fellow.

\section{The SSR under stochastic volatility}\label{sec:stochvol}

Our objective in this paper will be to compute the SSR under stochastic volatility.  For concreteness, let's focus on forward variance models of the form
\begin{align}
\frac{dS_t}{S_t }&= \sqrt{V_t}\,dZ_t\notag\\
d\xi_t(u) &= f_t(\xi)\,\kappa(u-t)\,dW_t,
\label{eq:SVmodel}
\end{align}
where $X = \log S$, $V_t\,dt = d\angl{X}_t$, and $d \angl{Z,W}_t = \rho\,dt$.  In particular, such a model is scale-invariant, with $\xi$ adapted to the filtration generated by $W$.

Assuming that we can compute option prices in model \eqref{eq:SVmodel}, computing $\cS_t(T)$ is in principle straightforward.  
If, as is usual in stochastic volatility modeling, we assume that $\xi$ is adapted to the filtration generated by $W$, the ATM volatility $\sigma_t(T)$ can only depend on the volatility state variables, namely the forward variances $\left\{\xi_t(u):u>t\right\}$.  Then, with $X = \log S$, as in Equation (9.5) of \cite{bergomi2015stochastic},
\begin{align} \label{equ:beta}
\beta_t(T) = \frac{\eet{d \angl{\sigma(T),X}_t}}{\eet{d\angl{X}_t}}.
\end{align}
Applying \IF, denoting the Fr\'echet derivative by $\delta$,
\begin{align*}
d \angl{\sigma(T),X}_t &= \int_t^T\,du\,\frac{\delta \sigma_t(T)}{\delta \xi_t(u)} \,\,d \angl{\xi(u),X}_t \\
&=  \sqrt{V_t}\, \int_t^T\,du\,\frac{\delta \sigma_t(T)}{\delta \xi_t(u)} \,\rho\,f_t(\xi)\,\kappa(u-t)\,dt.
\end{align*}
This gives
\begin{align}
\beta_t(T) &= \frac \rho {\sqrt{V_t}}\,
 \int_t^T\,\frac{\delta \sigma_t(T)}{\delta \xi_t(u)} \,f_t(\xi)\,\kappa(u-t)\,du.
 \label{eq:beta}
\end{align}
In order to streamline notation, let us define the operator
\begin{align}
D^\xi_t := \frac{1}{\sqrt{V_t}}\,\int_t^T \,du\,f_t(\xi)\,\kappa(u-t)\,\frac{\delta }{\delta \xi_t(u)} .
\label{eq:DxiDef}
\end{align}
With this notation \eqref{eq:beta} becomes
\begin{align}
\beta_t(T) &=  \rho\,D^\xi_t \sigma_t(T).
 \label{eq:beta2}
\end{align}

\begin{remark}
The right-hand side of \eqref{equ:beta} is really a short-hand for 
\[ \beta_t (T) = \frac{\tfrac{d}{d t} \langle \sigma (T), X
   \rangle_t}{\tfrac{d}{d t} \langle X \rangle_t}, \]
in terms of the usual quadratic covariation bracket of stochastic analysis. Semi-martingality of 
implied volatility, as seen from time $t$, with maturity $T$, as process in $t$, was discussed by various authors, 
notably \cite{durrleman2010implied}.%
\footnote{Strictly speaking, for constant money strike $K$. The forward ATM case of interest to us, amounts to take 
$K = S_t$ in which case semimartingality can be obtained via the It\^o-Wentzell formula.}
\end{remark}

\section{The SSR in terms of the characteristic function}\label{sec:charfn}

Let $\Sigma_t(k,T)=\sigma_{BS,t}(k,T)^2\,(T-t)$.  
From Equation (5.7) of \cite{gatheral2006volatility}, which is a straightforward consequence of the Lewis formula \cite{lewis2000option}, we have the following relationship between $\Sigma_t(k,T)$ and the characteristic function:
\beq
\int_{\R^+}\,\frac{da}{a^2+\tfrac14}\,\Re\left[e^{-\ui\,a\, k}\,e^{-\left(a^2+\tfrac14\right)\,\Sigma_t(k,T)}\right]=
\int_{\R^+}\,\frac{da}{a^2+\tfrac14}\,\Re\left[e^{-\ui\,a\, k}\varphi_t(T;a-\ui/2)\right].
\label{eq:Sigma}
\eeq
\eqref{eq:Sigma} leads to the following representation of the SSR $\cR_t(T)$ in terms of the characteristic function.

\begin{prop} \label{prof:SSRandCF}
Let $\psi= \log \varphi$.  Then
\begin{equation}
\cR_t(T) =-\frac{\int_{\R^+}\,\frac{da}{a^2+\tfrac14}\,\Re\left[\rho\,D^\xi_t \psi_t(T;a-\ui/2)\,\exp\left\{\psi_t(T;a-\ui/2)\right\}\right]}
{\int_{\R^+}\,\frac{a\,da}{a^2+1/4}\,\Im\left[\exp\left\{\psi_t(T;a-\ui/2)\right\}\right]}.
\label{eq:RCh}
\end{equation}
\end{prop}

\begin{proof}
%\subsection{Computation of $\cS_t(T)$}
Differentiating \eqref{eq:Sigma} wrt $k$, and performing the integration on the LHS, we obtain (5.8) of \cite{gatheral2006volatility}:
\beq
\cS_t(T)= - e^{\Sigma_t(0,T)/8}\sqrt{\frac{2}{\pi}}\frac{1}{\sqrt{T-t}}\,\int_{\R^+}\,\frac{a\,da}{a^2+1/4}\,\Im\left[\varphi_t(T;a-\ui/2)\right].
\label{eq:psiCh}
\eeq
Now we have the skew, we need only compute the regression coefficient $\beta_t(T)$.
Evaluating the functional derivative of \eqref{eq:Sigma} with respect to $\xi_t(u)$ at $k=0$,
\begin{align*}
\int_{\R^+}\,da\,\Re\left[\frac{\delta \Sigma_t(0,T)}{\delta \xi_t(u)} \,e^{-\tfrac12 \left(a^2+\tfrac14\right)\,\Sigma_t(0,T)}\right]=
\int_{\R^+}\,\frac{da}{a^2+\tfrac14}\,\Re\left[\frac{\delta}{\delta \xi_t(u)} \varphi_t(T;a-\ui/2)\right].
\end{align*}
The LHS may once again be integrated explicitly to give
\beas
\frac{\sqrt{2 \pi }}{\sqrt{\Sigma_t(0,T)} }\,\frac{\delta \Sigma_t(0,T)}{\delta \xi_t(u)} =e^{\tfrac18\,\Sigma_t(0,T)}\,
\int_{\R^+}\,\frac{da}{a^2+\tfrac14}\,\Re\left[\frac{\delta}{\delta \xi_t(u)} \varphi_t(T;a-\ui/2)\right].
\eeas
Equivalently,
\beas
\frac{\delta \sigma_t(T)}{\delta \xi_t(u)} = e^{\tfrac18\,\Sigma_t(0,T)}\,\sqrt{\frac{2}{\pi}}\,\frac{1}{\sqrt{T-t}}\,
\int_{\R^+}\,\frac{da}{a^2+\tfrac14}\,\Re\left[\frac{\delta }{\delta \xi_t(u)} \varphi_t(T;a-\ui/2)\right].
\eeas
Substituting into \eqref{eq:beta}, we obtain
\begin{align}
\beta_t(T) &=\frac \rho {\sqrt{V_t}}\,
 \int_t^T\,\frac{\delta \sigma_t(T)}{\delta \xi_t(u)} \,f_t(\xi)\,\kappa(u-t)\,du\notag\\
 &= \rho \,e^{\tfrac18\,\Sigma_t(0,T)}\,\sqrt{\frac{2}{\pi}}\,\frac{1}{\sqrt{T-t}}\,
\int_{\R^+}\,\frac{da}{a^2+\tfrac14}\,\Re\left[ D^\xi_t\,\varphi_t(T;a-\ui/2)\right].
\label{eq:betaCh}
\end{align}
Substituting from \eqref{eq:psiCh} and \eqref{eq:betaCh}, we obtain 
\beas
\cR_t(T) = -\frac{\int_{\R^+}\,\frac{da}{a^2+\tfrac14}\,  \rho\,\Re\left[D^\xi_t \varphi_t(T;a-\ui/2)\right]}
{\int_{\R^+}\,\frac{a\,da}{a^2+1/4}\,\Im\left[\varphi_t(T;a-\ui/2)\right]}.
\eeas
With $\varphi = e^\psi$, the result follows.

\end{proof}

\section{The SSR in affine forward variance (AFV) models}\label{sec:AFV}

Recall from \cite{gatheral2019affine} that in AFV models,
\beq
d\xi_t(u) = \kappa(u-t)\,\sqrt{V_t}\,dW_t,
\label{eq:AFV}
\eeq
so $f_t(\xi) = \sqrt{V_t}$.  The following lemma gives a formula for the SSR in AFV models.

\begin{prop} \label{prop:SSRaffine}
In an affine forward variance model of the form \eqref{eq:AFV},
\begin{align}
\cR_t(T) =- \frac{\int_{\R^+}\,\frac{da}{a^2+\tfrac14}\,\Re\left[\rho\,(\kappa \star g)(T-t;a-\ui/2)\,e^{\int_t^T\,\xi_t(s)\,g(T-s;a-\ui/2)\,ds}\right]}{\int_{\R^+}\,\frac{a\,da}{a^2+1/4}\,\Im\left[e^{\int_t^T\,\xi_t(s)\,g(T-s;a-\ui/2)\,ds}\right]},
\label{eq:RAFV}
\end{align}
where $g$ satisfies the convolution integral Riccati equation
\beq
g(\tau;a) =  -\frac 12\,a(a+\ui) +\ui\, \rho \,a\, (\kappa \star g) (\tau;a) + \frac 12\, (\kappa \star g) (\tau;a)^2.
\label{eq:Riccati}
\eeq

\end{prop}

\begin{proof}

From \cite{gatheral2019affine}, 
\beq
\psi_t(T;a) = \log \varphi_t(T;a) =  \int_t^T\,\xi_t(s)\,g(T-s;a)\,ds,
\label{eq:varphi}
\eeq
where $g$ satisfies the convolution integral Riccati equation \eqref{eq:Riccati}.
Functionally differentiating \eqref{eq:varphi} gives, for $u \in [t,T]$,
\begin{align*}
\frac{\delta }{\delta \xi_t(u)} \psi_t(T;a)&= g(T-u;a).
\end{align*}
Thus
\begin{align*}
D^\xi_t \psi_t(T;a-\ui/2)
&=  (\kappa \star g)(T-t;a-\ui/2).
\end{align*}
Substitution into \eqref{eq:RCh} yields the result.

\end{proof}

Given a (numerical or analytical) expression $g(\cdot)$ for the solution of the convolution Riccati equation \eqref{eq:Riccati}, the expression \eqref{eq:RAFV} may be evaluated numerically.  In particular, $\cR_t(T)$ may be evaluated in the rough Heston model.

\subsection{The Classical Heston SSR}

Consider the classical Heston model given by
\begin{align*}
\frac{dS_t}{S_t}&=\sqrt{V_t}\,dZ_t\\
dV_t &= -\lambda\,(V_t-\bar V)\,dt+\nu\,\sqrt{V_t}\,dW_t,
\end{align*}
with $d\angl{W,Z}_t = \rho\,dt$.  As is well-known (see \cite{gatheral2006volatility} for example), the characteristic exponent is given by
\begin{align}
\psi_t(T;a) =\log \eet{e^{\ui a X_T}}
 = \ui a\,X_t + D(\tau; a)\,V_t + C(\tau; a)\, \bar V,
 \label{eq:HestonCF}
 \end{align}
 where $X= \log S$, $\tau = T-t$, and $C$ and $D$ are complicated functions given in closed-form. The proof of the following lemma shows how the SSR may be computed in a Markovian model.
 
 \begin{lemma}
 
 In the classical Heston model,
  \begin{align}
 \cR_t(T) =\frac{
 \int_{\R^+}\,\frac{da}{a^2+1/4}\,\Re\left[\rho \nu\,D(\tau; a-\ui/2)\,
 \exp\left\{D(\tau; a-\ui/2)\,V_t + C(\tau; a-\ui/2)\, \bar V\right\} 
 \right]
 }{
 \int_{\R^+}\,\frac{a\,da}{a^2+1/4}\,\Im\left[ \exp\left\{D(\tau; a-\ui/2)\,V_t + C(\tau; a-\ui/2)\, \bar V\right\}\right]
 }
 \label{eq:ssrHeston}
 \end{align}

 \end{lemma}
 
 \begin{proof}

The Lewis formula \cite{lewis2000option}, gives the price of a European call option  in terms of the characteristic function by
 \begin{align}
 C_t = \eet{(S_T-K)^+} = S - \sqrt{S K}\,\frac{1}{\pi}\,\int_{\R^+}\,\frac{du}{u^2+1/4}\,\Re\left[\varphi_t(T;u-\ui/2)\right].
\label{eq:cHeston}
 \end{align}
 Following Section 9.8 on page 368 of \cite{bergomi2015stochastic}, $\beta_t(T)$ may be computed numerically by shifting the initial instantaneous  variance $V_t$. In infinitesimal terms, and with the notation of Section \ref{sec:intro}, Bergomi's recipe reads
 \begin{align*}
 \beta_t(T) 
 &=\rho\,\nu\,\left(\frac{\p C_{BS,t}}{\p \sigma_{\BS}}\right)^{-1}\,\frac{\p C_t}{\p V_t}.
 \end{align*}
 Substituting the explicit form \eqref{eq:HestonCF} of the characteristic function into \eqref{eq:cHeston}, and differentiating , we find (with $K=S_t$, taken without loss of generality as equal to $1$),
 \begin{align*}
 \frac{\p C_t}{\p V_t} = -\frac 1 \pi \int_{\R^+}\,\frac{du}{u^2+1/4}\,\Re\left[D(\tau; a)\,
 \exp\left\{ D(\tau; a)\,V_t + C(\tau; a)\, \bar V\right\}
 \right].
 \end{align*}
 Differentiating the Black-Scholes formula, again with $K=S_t=1$, we obtain that
 \begin{align*}
\left. \frac{\p C_{\BS,t}}{\p \sigma_{\BS}}  \right|_{ \sigma_{\BS}=\sigma_{t}(T)}
 =e^{-\Sigma_t(0,T)/8}\sqrt{\frac{T-t}{2\,\pi}},
 \end{align*}
using the familiar formula for the Black-Scholes Vega.  Thus
 \begin{align}
 \beta_t(T)&= -e^{\Sigma_t(0,T)/8}\,\rho\,\nu\,\sqrt{\frac{2}{\pi}} \,\sqrt{\frac{1}{T-t}} \notag\\
&\qquad\qquad\qquad \times \int_{\R^+}\,\frac{du}{u^2+1/4}\,\Re\left[D(\tau; a)\,
 \exp\left\{D(\tau; a)\,V_t + C(\tau; a)\, \bar V\right\}
 \right].
 \label{eq:betaHeston}
 \end{align}
 Also from  \eqref{eq:psiCh},
 \begin{equation}
\cS_t(T)= - e^{\Sigma_t(0,T)/8}\sqrt{\frac{2}{\pi}}\frac{1}{\sqrt{T-t}}\,\int_{\R^+}\,\frac{u\,du}{u^2+1/4}\,\Im\left[ \exp\left\{ D(\tau; a)\,V_t + C(\tau; a)\, \bar V\right\}\right].
\label{eq:skewHeston}
 \end{equation}
 Dividing \eqref{eq:betaHeston} by \eqref{eq:skewHeston} yields the result.
 
 \end{proof}

 Since $C(\tau; a)$ and $D(\tau; a)$ have closed-form expressions, $\cR_t(T)$ may be computed using straightforward numerical integration.  \eqref{eq:ssrHeston} is of course a special case of formula \eqref{eq:RAFV}; whereas Bergomi's method works only for Markovian models, formula \eqref{eq:RAFV} applies to all affine forward variance models.

\section{$\cR_t(T)$ from the forest expansion}\label{sec:forest}

From~\cite{alos2020exponentiation, friz2022forests}, we have the following definition of the diamond product.
\begin{definition} \label{def:diamond}
Given two continuous semimartingales $A,B$ with integrable covariation process $\langle A , B \rangle$, the diamond product is defined by
$$
     (A \dm B)_t (T) := \eef{\langle A , B \rangle_{t,T}} = \eef{\langle A , B \rangle_{T}} -\langle A , B \rangle_t \; .
$$
\end{definition} 

Diamond products of (sufficiently integrable, continuous) semimartingales 
naturally lead to binary ``diamond'' trees 
such as $(A \dm B) \dm C$.  
Diagrammatically, the diamond product of two trees $\mT_1$ and $\mT_2$ is represented by {\it root joining}, 
\tkz
$$
\mT_1 \dm \mT_2 = \T1T2d,
$$
where the two binary trees $\mT_1$ and $\mT_2$ are represented as the single leaves $\S$ and $\U$. We regard linear combinations of diamond trees as {\em forests}.   

From \cite{alos2020exponentiation}, the cumulant  generating function (CGF) is given by the forest expansion 
\beq
\psi_t(T;a) := \log \varphi_t(T;a) =  \sum_{\ell=0}^\infty\,\tilde {\mF}_\ell(a),
\label{eq:CGFmF}
\eeq
where the $\tmF_\ell$ satisfy the recursion
$
\tmF_0 =  -\tfrac 1 2 a(a+\ui)\,M
$
and for $\ell>0$,
\beq
\tmF_\ell=\frac12\,\sum_{j=0}^{\ell-2}\,\left(\tmF_{\ell-2-j} \dm \tmF_j\right) +\ui a\, \left(X \dm \tmF_{\ell-1}\right).
\label{eq:tmFrecursion}
\eeq

\noindent Applying the recursion \eqref{eq:tmFrecursion}, the first few $ \tmF$ forests are given by
\tikzexternaldisable
\beas
\tmF_0  &=&-\tfrac 1 2 a(a+\ui)\, \M  \nonumber\\
\tmF_1 &=& -\tfrac \ui 2 a^2(a+\ui)\, \MXd \nonumber\\
\tmF_2 &=& \tfrac1{2^3} a^2\,(a+\ui)^2\,\MMd+ 
\tfrac 12 a^3\,(a+\ui)\,  \MXdXd
  \nonumber\\
\tmF_3 &=&  (\tmF_0 \dm \tmF_1) + \ui a\,\X \dm \tmF_2 \nonumber\\
&=& \tfrac{\ui}{2^2} a^3\,(a+\ui)^2\,\XMdMd
+ \tfrac{\ui}{2^3} a^3\,(a+\ui)^2\,\MMdXd+\tfrac{\ui}{2} a^4\,(a+\ui)\,\MXdXdXd,
\label{eq:firstFive}
\eeas
\tikzexternaldisable
where $\X = X_t$ and $\M=M_t(T) = (X \dm X)_t(T) = \int_t^T\,\xi_t(u)\,du$.

\begin{remark}
Note that the total probability and martingale constraints are satisfied for each tree.
That is $\psi_t^T(0) = \psi_t^T(-\ui) = 0$.
\end{remark}

\subsection{Computation of $\cR_t(T)$ to second order in the forest expansion}

Consider a formal expansion according to values of $\ell$, which we term the {\em forest expansion}.
The number of leaves in each tree in the forest  $\tmF_\ell$ is given by $\ell+2$  where $\X$ counts as one leaf, and $\M = \XXd$ counts as two leaves.

\begin{prop}\label{prop:2ndOrder}
To second order in the forest expansion,
\begin{align}
\cR_t(T)
  %%%%%%%%%%%%%%%%%%%%%%%%%%
 &= \frac{M_t(T)\,\rho\,
D^\xi_t \left\{\M +\tfrac 12\,
 \MXd  
 \right\}}
 %%%%%%%%%%%%%%%%%%%%%%%%%%
 {
 \MXd  +\MXdXd
 }.
 \label{eq:SSRnextOrder}
\end{align}

\end{prop}

\begin{proof}
From \eqref{eq:CGFmF}, to second order in the forest expansion,
\beas
\psi_t(T;a-\ui/2) =  -\tfrac12 \left(a^2 +\tfrac14\right)\left\{\M +(\ui\,a+\tfrac 12)\,
 \MXd -\tfrac1{4} \left(a^2+
 \tfrac1{4}\right)\,\MMd-
  (a-\ui/2)^2\,  \MXdXd\right\}.
\eeas
Then
 \begin{align*}
\Re\left[e^{\psi_t(T;a-\ui/2)}\,D^\xi_t \left\{\M +(\ui\,a+\tfrac 12)\,
 \MXd 
 \right\}\right]
 &\approx \Re\left[e^{-\tfrac1{2} \left(a^2+\tfrac14\right)\,M}\,D^\xi_t \left\{\M +(\ui\,a+\tfrac 12)\,
 \MXd
 \right\}\right]\\
 &=e^{-\tfrac1{2} \left(a^2+\tfrac14\right)\,M}\,D^\xi_t\left\{\M +\tfrac 12\,
 \MXd  \right\},
 \end{align*}
and
 \begin{align*}
\Im\left[e^{\psi_t(T;a-\ui/2)}\right]
 &=-\tfrac12\, e^{-\tfrac1{2} \left(a^2+\tfrac14\right)\,M}\,a\,\left(a^2+\tfrac14\right)\,\left\{
\MXd + \MXdXd 
 \right\}.
 \end{align*}
 The result follows by substitution into \eqref{eq:RCh}.
 \end{proof}

\subsection{Computation of $\cR_t(T)$ to leading order}% in $\tau$}

\begin{lemma}\label{lem:leadingorder}
To leading order, %writing the previous result out in full,
\begin{align}
\cR_t(T)
  %%%%%%%%%%%%%%%%%%%%%%%%%%
 &= \frac{M_t(T)}{\sqrt{V_t}}\,\frac{\int_t^T  \,f_t(\xi)\,\kappa(u-t)\,du}
 {\int_t^T \,ds\,\int_s^T\,\eet{\sqrt{V_s}  \,f_s(\xi)}\,\kappa(u-s)\,du}.
 \label{eq:SSRleadingOrder2}
\end{align}
\end{lemma}

\begin{proof}

From Proposition \ref{prop:2ndOrder}, to leading order,
\begin{align*}
\cR_t(T)
  %%%%%%%%%%%%%%%%%%%%%%%%%%
 &= \frac{M_t(T)\,\rho\,
D^\xi_t \M }
 %%%%%%%%%%%%%%%%%%%%%%%%%%
 {
 \MXd 
 %-\tfrac14\,(\tfrac 3M+\tfrac14)\,\MXd^2
 }.
% \label{eq:SSRnextOrder}
\end{align*}
Since $\frac{\delta }{\delta \xi_t(u)} M_t(T) = 1$, by the definition \eqref{eq:DxiDef} of $D^\xi$,
\begin{align*}
D^\xi_t \M = \frac{1}{\sqrt{V_t}}\,\int_t^T  \,f_t(\xi)\,\kappa(u-t)\,du,
%\label{eq:DxiM}
\end{align*}
Also
 \begin{align*}
 \MXd =  \rho\,\int_t^T \,ds\,\left(\int_s^T\,\eet{\sqrt{V_s}  \,f_s(\xi)}\,\kappa(u-s)\,du\right).%\label{eq:MXd}
 \end{align*}

\end{proof}

The following corollary gives the limit as $T \to t$ of the SSR for any stochastic volatility model of the form \eqref{eq:SVmodel}.
\begin{cor}
Let $\tau = T-t$.  Then
\begin{align}
\lim_{T \to t} \cR_t(T) 
&=  \tau\,\frac{d}{d\tau} \log \left(\int_0^\tau \,ds\,\int_0^{s}\,\kappa(u)\,du\right).
\label{eq:Rlod}
\end{align}

\end{cor}

\begin{proof}

We have $\xi_t(u) = V_t + \text{higher order}$ and
$$
\eet{\sqrt{V_s}  \,f_s(\xi)} = \sqrt{V_t}  \,f_t(\xi) + \text{higher order}.
$$
Thus   to leading order as $T \to t$, $M_t(T)\sim V_t\, \tau$ and from \eqref{eq:SSRleadingOrder2},
\begin{align*}
\cR_t(T)
 &= \frac{M_t(T)}{V_t}\,\frac{\int_t^T  \,\kappa(u-t)\,du}
 {\int_t^T \,ds\,\left(\int_s^T\,\kappa(u-s)\,du\right)}
 = \tau\,\frac{d}{d\tau} \log \left(\int_0^\tau \,ds\,\int_0^{s}\,\kappa(u)\,du\right).
\end{align*}

\end{proof}

The following straightforward corollary confirms a formal computation of Fukasawa in Remark 2.10 of \cite{fukasawa2021volatility}.

\begin{cor}\label{cor:alphaplus1}
Let  $
\kappa(s)\sim s^{\alpha-1}
$ as $s \to 0$.  
%\,L_\kappa(s)
%$ where $L_\kappa$ is a slowly varying function.  
Then
\begin{align*}
\lim_{t \uparrow T} \cR_t(T)  = \alpha + 1.
\end{align*}
\end{cor}

\begin{remark} In terms of the Hurst parameter $H > 0$, one has $\alpha \equiv H +1/2$, so that in the $t \uparrow T$ limit,
$\cR_t(T) \sim 3/2 + H$, confirming a formal computation of Fukasawa in Remark 2.10 of
 \cite{fukasawa2021volatility}.
In particular, one has, when $T \to t$ and in diffusive situations with $H=2$, skew-stickiness-ratio $2$, cf. Chapter 9 in \cite{bergomi2015stochastic}.

These (asymptotic) results are related to various ``rules'' that relate local and implied volatility skews. In a diffusive setting the rule that the local volatility skew is double the implied volatility skew has been known for a long time \cite{gatheral2006volatility}; this was recently extended to a $H+3/2$-rule in the rough volatility regime $H<1/2$, using a combination of Malliavin calculus and rough analysis \cite{bourgey2023local}.
To appreciate this relation, it suffices to note that asymptotically local variance is the expectation of instantaneous variance conditional on a small move in the underlying (e.g. \cite{demarco2018local} and references therein) and that the short term limit of implied variance is instantaneous variance \cite{durrleman2010implied}, -- precisely consistent with the definition of $\beta_t(T)$.

\end{remark}

%%%%%%%%%%%%%%%%%%%%%%%%%%%%%%%%%%

\subsection{Computation of $\cR_t(T)$ to next-to-leading order in $\tau$}

The forest expansion \eqref{eq:CGFmF} is effectively a small $\nu$ (vol-of-vol) expansion.  As noted in Chapter 9 of \cite{bayer2023rough}, with $\tau - T-t$ and for fixed $a$, the $\tmF_\ell(\tau)$ scale as $\tau^{\ell\,\alpha+1}$  as $\tau \downarrow 0$. When computing $\cR_t(T)$ using \eqref{eq:RCh}, different powers of $a$ will, after integration, generate different powers of $\tau$.
Thus, to get a small $\tau$ expansion of the SSR, it is not sufficient to go to second order in the forest expansion; we also need to take powers of $a$ into account.  %\end{remark}

 $a^n\,\mF^\ell$, when integrated over the positive real line with respect to the Gaussian kernel $e^{-\tfrac1{2} \left(a^2+\tfrac14\right)\,M}$, scales as $\tau^{-(n+1)/2}\,\tau^{\ell\,\alpha+1} = \tau^{\ell\,\alpha-(n-1)/2} $.  Thus, we get the following table of examples of relative scaling of typical terms (dividing by $\sqrt{\tau}$ ):
 
\begin{center}
\begin{tabular}{|c|c|}
\hline
$\M$ &$1$\\
$\MXd$& $\tau^{\alpha}$\\
$\MXdXd$& $\tau^{2 \alpha}$\\
$a^2\,\MMd$& $\tau^{2\alpha-1}$\\
$a^2\,\MXdXd$& $\tau^{2\alpha-1}$\\
$a^2\,(\MXd)^2$& $\tau^{2\alpha}$\\
$a^2\,\M\,\MXd$& $\tau^{\alpha}$\\
$a^4\,(\MXd)^2$& $\tau^{2\alpha-1}$\\
$a^6 \MXd^3$& $\tau^{3\alpha-1}$\\
\hline
\end{tabular}
\end{center}

\noindent Taking the above scaling carefully into account yields the following proposition.

\begin{prop}\label{prop:42}

To next-to-leading order in $\tau=T-t$,
\begin{small}
\begin{align}%\boxed{
&\cR_t(T) =\rho\,M\,\bigg\lbrace \notag\\
& \frac{D^\xi_t
 \M +\tfrac 12\, D^\xi_t\MXd   - \tfrac1{4\,M}\,D^\xi_t\MMd-\frac{1}{M}\, D^\xi_t \MXdXd
 -\tfrac1{4\,M}\,\MXd \,D^\xi_t\M+\tfrac 3{2\,M^2}\,\MXd\,D^\xi_t\MXd  
 + \frac{3}{8 M^2}  D^\xi \M\,\MMd
 +\frac{3}{2 M^2}  D^\xi \M\,\MXdXd  
 -\frac{15}{8 M^3} D^\xi \M \MXd^2
 +...
 }
{
\MXd   + \MXdXd -\tfrac 3{4\,M}\,(\MXd)^2
-\frac{105}{24 M^3} \MXd^3+\frac{15}{8 M^2} \MMd\,\MXd+\frac{15}{2 M^2} \MXd\,\MXdXd-\frac{3}{4 M} \MMdXd-\frac{3 }{2 M}\XMdMd-\frac{3}{M} \MXdXdXd
%+\tfrac{15}{16\,M^2}\,\MXd \,\MMd -\tfrac{15}{2\,M^2}\,\MXd\,\MXdXd 
+...
 }\bigg \rbrace,\notag\\
\label{eq:Rforest1}
\end{align}
\end{small}
where $+...$ denotes higher order terms.

\end{prop}

\begin{proof}

To consistently expand to next-to-leading order in $\tau$, the highest order forest that can contribute is $\tilde \mF_3$.
From \eqref{eq:CGFmF}, to third order in the forest expansion,
\begin{align*}
\psi_t(T;a-\ui/2) &=  -\tfrac12 \left(a^2 +\tfrac14\right)
\big\lbrace \M +(\ui\,a+\tfrac 12)\,\MXd
 - \tfrac1{4} \left(a^2+ \tfrac1{4}\right)\,\MMd-
  (a-\ui/2)^2\,  \MXdXd\\
  %%%%%%%%%%%%%%%%%%%%%%%%%%%%
 &\qquad\qquad\qquad -\tfrac{\ui}{2} (a^2+\tfrac14)\,(a-\ui/2)\,\XMdMd
- \tfrac{\ui}{4} \left(a^2+ \tfrac1{4}\right)\,(a-\ui/2)\,\MMdXd
-\ui\, (a-\ui/2)^3\,\MXdXdXd
\big\rbrace.
\end{align*}
For the numerator, to next-to-leading order in $\tau$, we obtain
\begin{align*}
&\Re\left[e^{\psi_t(T;a-\ui/2)}\,D^\xi_t \psi_t(T;a-\ui/2)\right]=-\frac{1}{2} \left(a^2+\frac{1}{4}\right)\,e^{-\tfrac1{2} \left(a^2+\tfrac14\right)\,M}\,\,\bigg\lbrace\\
 &\qquad \qquad\qquad \qquad \qquad \qquad D^\xi_t
 \M +\tfrac 12\, D^\xi_t\MXd   
 -\frac{a^2}4 D^\xi \MMd
  -a^2 D^\xi \,\MXdXd
  -\frac{1}{4} a^2 D^\xi \M\,\MXd \\
&\qquad \qquad\qquad \qquad \qquad \qquad  + \frac{1}{2} a^4 D^\xi \MXd\,\MXd+ \frac{1}{8} a^4 D^\xi \M\,\MMd
 +\frac{1}{2} a^4 D^\xi \M\,\MXdXd  
 -\frac18 a^6 D^\xi \M \MXd^2+...
\bigg \rbrace.
 \end{align*}

 \noindent Then we do the same for the denominator, including terms up to $\tau^{2\alpha+1}$ or $\tau^{3\,\alpha}$.
 \begin{align*}
&\Im\left[e^{\psi_t(T;a-\ui/2)}\right]=-\frac{1}{2} a \left(a^2+\frac{1}{4}\right)\,e^{-\tfrac1{2} \left(a^2+\tfrac14\right)\,M}\,\,\bigg\lbrace\\
 &\qquad\qquad \MXd   + \MXdXd -\tfrac14\,a^2\,(\MXd)^2 
-\frac{1}{24} a^6 \MXd^3+\frac{1}{8} a^4 \MMd\,\MXd+\frac{1}{2} a^4 \MXd\,\MXdXd-\frac{a^2}{4} \MMdXd-\frac{a^2 }{2}\XMdMd-a^2 \MXdXdXd +...
\bigg \rbrace.
 \end{align*}
Performing the integrations in \eqref{eq:RCh} explicitly yields the result.

\end{proof}

%\subsubsection{A useful lemma}

Before proceeding with examples of the application of Proposition \ref{prop:42}, we present a useful lemma that applies only to AFV models, and that incidentally justifies our choice of definition of the operator $D^\xi$.

\begin{lemma}\label{lem:DxiT}
Let $\mT$ be a diamond tree in an affine forward variance (AFV) model.  Then
\begin{align*}
(\mT \dm X)_t(T) = \rho\,\int_t^T\,\xi_t(s)\,D^\xi_s \mT_s(T)\,ds.
\end{align*}
\end{lemma}

\begin{proof}

From Lemma 4.3 of \cite{friz2022forests}, all diamond trees in an AFV model take the form
\begin{align*}
 \mT_t(T) = \int_t^T\,\xi_t(u)\,h(T-u)\,du,
\end{align*}
for some integrable function $h$.  Then by the definition \eqref{eq:DxiDef} of $D^\xi_t$,
\begin{align*}
D^\xi_t \mT_t(T) = \int_t^T\,\kappa(u-t)\,h(T-u)\,du.
\end{align*}
On the other hand,
\begin{align*}
(\mT \dm X)_t(T) = \rho\,\int_t^T\,\xi_t(s)\,ds\,\int_s^T\,\kappa(u-s)\,h(T-u)\,du.
\end{align*}
%So
%\begin{align*}
%-\p_t (\mT \dm X)_t(T) = \rho\,\xi_t(t)\int_t^T\,\kappa(u-s)\,h(T-u)\,du 
%= \rho\,V_t\,D^\xi_t \mT_t(T).
%\end{align*}

\end{proof}

\begin{remark}
It is straightforward to check, for example by examining explicit computations under rough Bergomi in Chapter 9 of \cite{bayer2023rough},  that Lemma \ref{lem:DxiT} does not hold in general for models that are not affine.
\end{remark}

\subsubsection{Example: Classical Heston with $\lambda=0$ and $\xi_t(u) = V$}

In this (probably simplest possible) case, we have
\begin{align*}
M=\M &= V\,\tau\\
\MXd &= \tfrac12\,\rho\,\nu\,V\,\tau^2\\
\MMd &= \tfrac13\,\nu^2\,V\,\tau^3\\
\MXdXd &= \tfrac16\,\rho^2\,\nu^2\,V\,\tau^3\\
\XMdMd &= \tfrac18\,\rho\,\nu^3\,V\,\tau^4\\
\MMdXd &=\tfrac1{12}\,\rho\,\nu^3\,V\,\tau^4\\
\MXdXdXd &=\tfrac1{24}\,\rho^3\,\nu^3\,V\,\tau^4.
\end{align*}
Also, since the forward variance curve is flat and the model is affine, $D^\xi \equiv \nu\,\p_V$ so
\begin{align*}
D^\xi \M &=\nu\,\tau\\
D^\xi  \MXd &= \tfrac12\,\rho\,\nu^2\,\tau^2\\
D^\xi \MMd &= \tfrac13\,\nu^3\,\tau^3\\
D^\xi  \MXdXd &= \tfrac16\,\rho^2\,\nu^3\,\tau^3.
\end{align*}

\noindent Substituting into \eqref{eq:Rforest1}, we obtain
%\begin{align*}
%\cR_t(T) =-\frac{\rho\,\nu\,V\,\tau^2\,\left\{
% 1+\tfrac 18\,\rho\,\nu\,\tau  - \tfrac1{12\,V} \,\nu^2\,\tau -\frac{1}{6\,V}\,\rho^2\,\nu^2\,\tau +\tfrac 3{4\,V}\,\rho^2\,\nu^2\,\tau
% \right\}}
%{
%\tfrac12\,\rho\,\nu\,V\,\tau^2   + \tfrac16\,\rho^2\,\nu^2\,V\,\tau^3 -\tfrac 3{16}\,\rho^2\,\nu^2\,V\,\tau^3
%%+\tfrac{15}{16\,M^2}\,\MXd \,\MMd -\tfrac{15}{2\,M^2}\,\MXd\,\MXdXd 
% }.
%\end{align*}
%
\begin{align*}
\cR_t(T) =2\,\frac{1+\tfrac{1 }{8}\,\nu  \rho  \tau+\tfrac 1 {24}\,\tfrac{\nu^2\,\tau}{V}- \tfrac{1}{96}\,\tfrac{\rho^2 \nu^2\tau}{V}+\cO(\tau^2)}
{
1-\tfrac 1{24}\,\rho\,\nu  \tau +\tfrac 1 8\,\tfrac{\nu^2\,\tau}{V}-\tfrac 3{32}\,\tfrac{\rho^2 \nu^2\,\tau}{V} +\cO(\tau^2) },
\end{align*}
which agrees exactly with the exact expansion of the closed-form classical solution.

%%%%%%%%%%%%%%%%%%%%%%%%%%%
\subsubsection{Example: Rough Heston with $\lambda=0$ and $\xi_t(u) = V$}

In the rough Heston case, we don't have a closed-form expression for the characteristic function but we do have rational approximations \cite{gatheral2019rational, gatheral2024generalization}.  The diamond trees required to implement \eqref{eq:Rforest1} are straightforward to compute and are given by (see page 193 of \cite{bayer2023rough} for example):
\begin{align*}
M=\M &= V\,\tau\\
\MXd &= \frac{\rho\,\nu}{\Gamma(2+\alpha)}\,V\,\tau^{\alpha+1}\\
\MMd &=\frac{\nu^2}{\Gamma(1+\alpha)^2}\,V\,\frac {\tau^{2 \alpha+1}}{2 \alpha+1}\,\\
\MXdXd &= \frac{\rho^2 \nu^2}{\Gamma(2+2 \alpha)}\,V\,\tau^{2 \alpha+1}\\
\XMdMd &= \frac{\rho\,\nu^3}{\Gamma(1+\alpha)\,\Gamma(1+2 \alpha)}\,V\,\frac{\tau^{3\alpha+1}}{3\alpha+1}\\
\MMdXd &=\frac{\rho\,\nu^3\,\Gamma(1+2\alpha)}{\Gamma(1+\alpha)^2\Gamma(2+3\alpha)}\,V\,\tau^{3\alpha+1}\\
%%%
\MXdXdXd &=\frac{\rho^3\,\nu^3}{\Gamma(2+3\alpha)}\,V\,\tau^{3\alpha+1}.
\end{align*}

\noindent Following the notation of \cite{bayer2023rough}, let
$$
I_t^{(j)}(T) = \int_t^T\,\xi_t(s)\,(T-s)^{j \alpha}\,ds.
$$
Then from the definition \eqref{eq:DxiDef} of $D^\xi_t$, with $\kappa(\tau) =  \frac{\nu\,\tau^{\alpha-1}}{\Gamma(\alpha)}$,
\begin{align*}
D^\xi_t I_t^{(j)}(T) &= \frac{1}{\sqrt{V_t}}\,\int_t^T \,du\,f_t(\xi)\,\kappa(u-t)\,\frac{\delta }{\delta \xi_t(u)} I_t^{(j)}(T)\\ &=  \frac{\nu}{\Gamma(\alpha)}\,\int_t^T \,(u-t)^{\alpha-1}\,(T-u)^{j \alpha}\,du\\
&= \frac{\Gamma (1+j \alpha)}{\Gamma (1+(j+1) \alpha)}\, \tau^{(j+1)\alpha }.
\end{align*}
It follows that
\begin{align*}
D^\xi \M &=\frac{\nu}{\Gamma(1+\alpha)}\,\tau^\alpha\\
D^\xi  \MXd &= \frac{\rho\,\nu^2}{\Gamma(1+2\alpha)}\,\tau^{2 \alpha}\\
D^\xi \MMd &=\frac{\nu^3}{\Gamma(1+\alpha)^2}\,\frac{\Gamma(1+2 \alpha)}{\Gamma(1+3 \alpha)}\,\tau^{3 \alpha}\\ 
D^\xi  \MXdXd &= \frac{\rho^2 \nu^3}{\Gamma(1+3 \alpha)}\,\tau^{3 \alpha}.
\end{align*}

\begin{remark} Note the general pattern, true only for AFV models when the forward variance curve is flat, whereby for a given tree $\mT$,
$\rho\,V_t\,D^\xi_t \mT_t(T) = \p_\tau \, (\mT \dm X)_t(T)
$.  %
This is a straightforward corollary of Lemma \ref{lem:DxiT}.
\end{remark}

In Figure \ref{fig:SSR15}, with $\xi = 0.025$, $\nu = 0.4$, and $\rho=-.8$, we plot the SSR for various values of $H$, numerically using \eqref{eq:RAFV} and using the forest formula \eqref{eq:Rforest1} with the above substitutions.  We observe clear consistency of the forest formula with the numerical computations.  However, in the practically relevant case of $H \lessapprox 0.2$, the timescale over which these two agree is arguably too short to be of any practical interest.

\begin{figure}[h!]
\begin{center}
\includegraphics[width = 1.0 \linewidth]{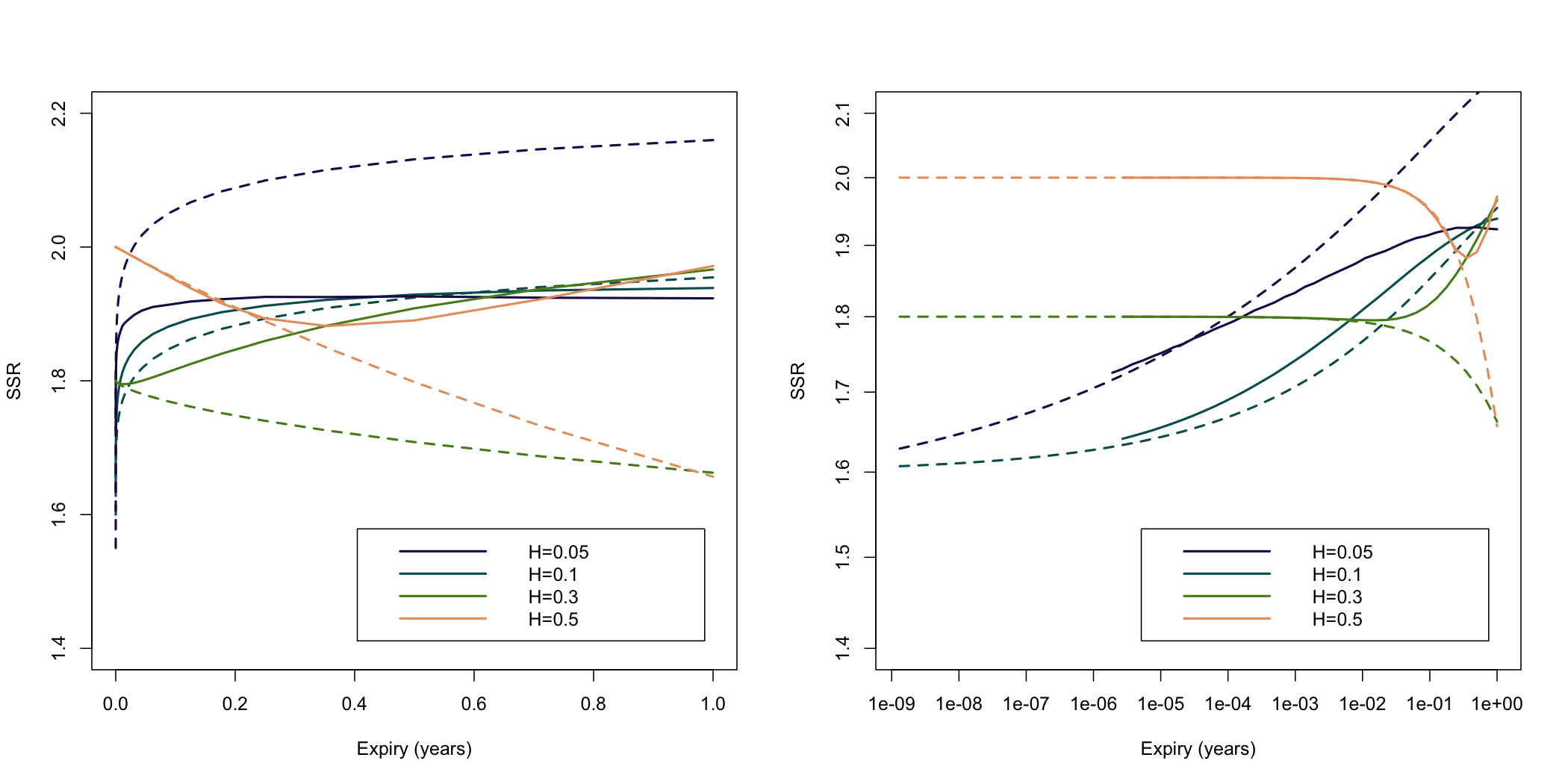}
\caption{The rough Heston SSR with a flat forward variance curve $\xi = 0.025$, and parameters $\nu = 0.4;\,\rho=-.8$.  Solid lines are the numerical computation \eqref{eq:RAFV} and dashed lines are the forest formula \eqref{eq:Rforest1}.  On the right is a log-log plot.  }
\label{fig:SSR15}
\end{center}
\end{figure}

%%%%%%%%%%%%%%%%%%%%%%%%%%%%%%%%%%%%%%%%%%%%%%

\section{Dependence of $\cR_t(T)$ on $\xi_t(u)$ and path-dependence}\label{sec:pathdep}

We may rewrite the leading order expression \eqref{eq:SSRleadingOrder2} suggestively in the form
\begin{align}
\cR_t(T)
  %%%%%%%%%%%%%%%%%%%%%%%%%%
 &= \frac{\left(\int_t^T\,\xi_t(s)\,ds\right)\,\int_t^T  \,\sqrt{V_t}\,f_t(\xi)\,\kappa(u-t)\,du}
 {V_t\,\int_t^T \,ds\,\int_s^T\,\eet{\sqrt{V_s}  \,f_s(\xi)}\,\kappa(u-s)\,du}.
 \label{eq:SSRsuggestive}
\end{align}
We first observe that \eqref{eq:SSRsuggestive} should be rather insensitive to the level of the forward variance curve.  On the other hand, $\cR_t(T)$ is clearly sensitive to the shape of $\xi_t(u)$.  A monotonic increasing forward variance curve will cause the SSR to increase relative to the flat curve case, and vice versa.

In the AFV case,  \eqref{eq:SSRsuggestive} simplifies to 
$$
\mathcal{R}_t(T)  
=\frac{\left(\int_t^T\,\xi_t(s)\,ds\right)\,\tilde \kappa(T-t)}{\int_t^T\,\xi_t(s)\,\tilde \kappa(T-s)\,ds},
%\label{eq:AFVSSR}
$$
where $\tilde \kappa(\tau) = \int_0^\tau\,\kappa(s)\,ds$.  If the forward variance curve is flat with $\xi_t(u) = \bar V$, the SSR does not depend on the level $V$ at all!
However, once again, $\cR_t(T)$ does depend on the shape of $\xi_t(u)$.

In Figure \ref{fig:SSRcrv}, with parameters $H=0.1;\,\nu = 0.4;\,\rho=-.8$ we plot the rough Heston SSR using the AFV formula \eqref{eq:RAFV}, and assuming three different forward variance curves of the form:
\begin{align*}
\xi_t(u) = (V_t - \bar V)\,e^{-\lambda t} + \bar V,
\end{align*}
with $\bar V = 0.025,\, \lambda = 7$ and $V_t = 0.025$ (flat), $V_t = 0.005$ (contango), and $V_t = 0.045$ (backwardation).  We see that in fact, the SSR is very sensitive to the shape of the forward variance curve.

\begin{figure}[h!]
\begin{center}
\includegraphics[width = 1.0 \linewidth]{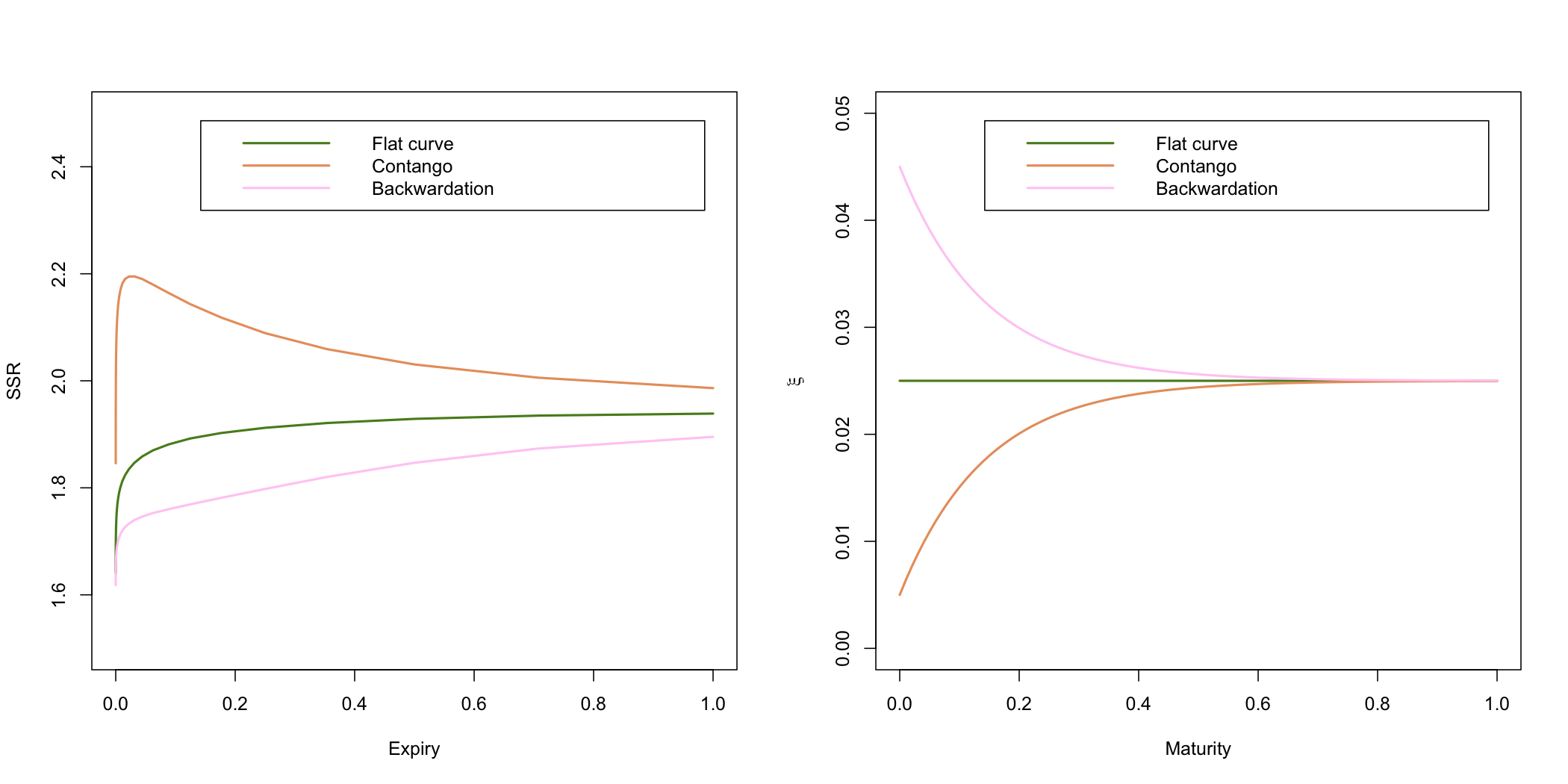}
\caption{The rough Heston SSR with various forward variance curves, and parameters $H=0.1;\,\nu = 0.4;\,\rho=-.8$.  The left-hand plot is of the SSRs and the right-hand plot shows the assumed forward variance curves.  }
\label{fig:SSRcrv}
\end{center}
\end{figure}

\subsection{Path-dependence of the SSR}

In an affine forward variance (AFV) model, with $dW_t = \rho\,dZ_t + \sqrt{1-\rho^2}\,dZ^\perp_t$, for each $L>0$, we may write
\begin{align*}
\xi_t(u) &=\xi_{-L}(u)+ \int_{-L}^t\,\kappa(u-r)\,\sqrt{V_r}\,dW_r.
\end{align*}
Formally taking the limit $L \to \infty$, we obtain
\begin{align*}
&=\bar \xi + \int_{-\infty}^t\,\kappa(u-r)\,\sqrt{V_r}\,dW_r\\
&=\bar \xi +\rho\, \int_{-\infty}^t\,\kappa(u-r)\,\frac{dS_r}{S_r} + \text{independent noise}.
\end{align*}
The forward variance curve depends on a weighted average of historical stock returns.  Indeed when the correlation $\rho = \pm 1$ (a purely path-dependent model in the terminology of \cite{guyon2023volatility}), the forward variance is a functional of this weighted average of historical stock returns.  We deduce that $\cR_t(T)$ depends on weighted average historical stock returns.   Specifically, if recent returns are very negative, we expect the forward variance curve to be backwardated, lowering the SSR, and vice versa.

This argument goes through for every forward variance model -- the forward variance curve is a noisy transform of the historical series of stock returns.

\section{Sensitivity of the SSR to model dynamics}\label{sec:sensitivity}

Recall from \cite{gatheral2019affine} that the rough Heston kernel takes the form
$$
\kappa(\tau) = \nu\,\tau^{\alpha-1}\,E_{\alpha,\alpha}\left(-\lambda\,\tau^\alpha\right),
$$
where $E_{\alpha,\alpha}(\cdot)$ is the Mittag-Leffler function and $\alpha=H+\tfrac12$.  
In order to demonstrate the sensitivity of the SSR to model dynamics, we choose three values of the parameter $\lambda\, (= 0,1,2)$ and find values of $H$ and $\nu$ such that the resulting parameter sets $\Pi_0$, $\Pi_1$ and $\Pi_2$ (listed in Table
\ref{tab:params}) generate the almost identical 1-month, 3-month, 6-month, and 12-month smiles shown in Figure \ref{fig:FourSmiles3}.

\begin{table}[h]
\centering
\begin{tabular}{|c|c|c|c|}
\hline
 & $\Pi_0$& $\Pi_1$ & $\Pi_2$ \\ \hline
$\lambda$ & $0$ & $1$ & $2$ \\ \hline
$H$ & $0.10$ & $0.223$ & $0.302$ \\ \hline
$\nu$ & $0.40$ & $0.481$ & $0.647$ \\ \hline
$\rho$ & $-0.65$ & $-0.65$ & $-0.65$ \\ \hline
\end{tabular}
\caption{Parameter values for the different parameter sets.}
\label{tab:params}
\end{table}

\begin{figure}[h!]
\begin{center}
\includegraphics[width = 0.95 \linewidth]{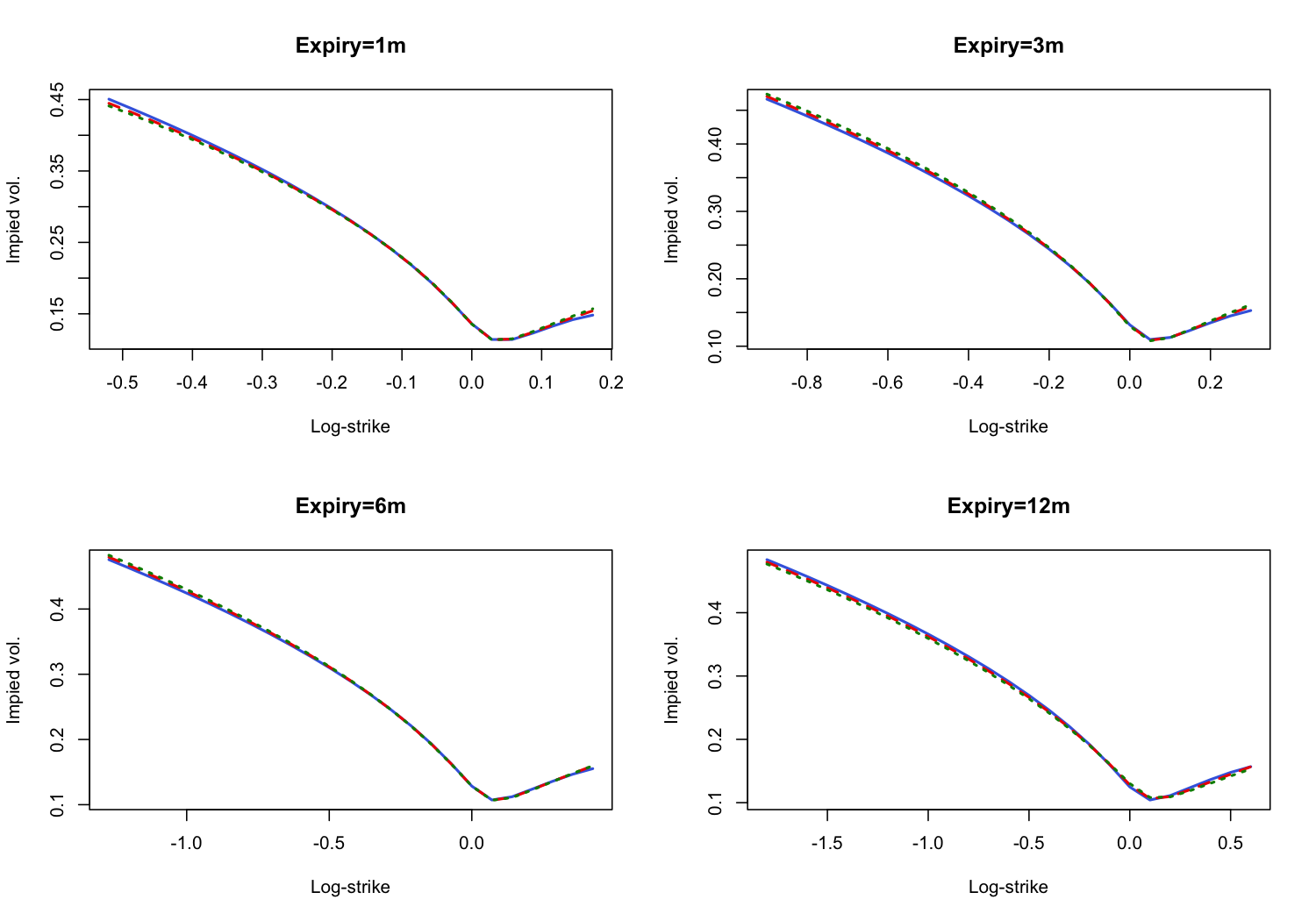}
\caption{Almost identical rough Heston smiles generated by the parameter sets $\Pi_0$, $\Pi_1$, and $\Pi_2$, in blue, red and green respectively.}
\label{fig:FourSmiles3}
\end{center}
\end{figure}

In Figure \ref{fig:SSR3}, we plot the SSR, computed using the affine SSR formula \eqref{eq:RAFV}, for each of these parameter sets as a function of time to expiry.  We see that although all three parameter sets generate very similar smiles, the SSR plots are very different.  We conclude that the SSR is highly senstive to dynamical assumptions; it cannot be deduced from the shape of the implied volatility surface.

\begin{figure}[h!]
\begin{center}
\includegraphics[width = 0.8 \linewidth]{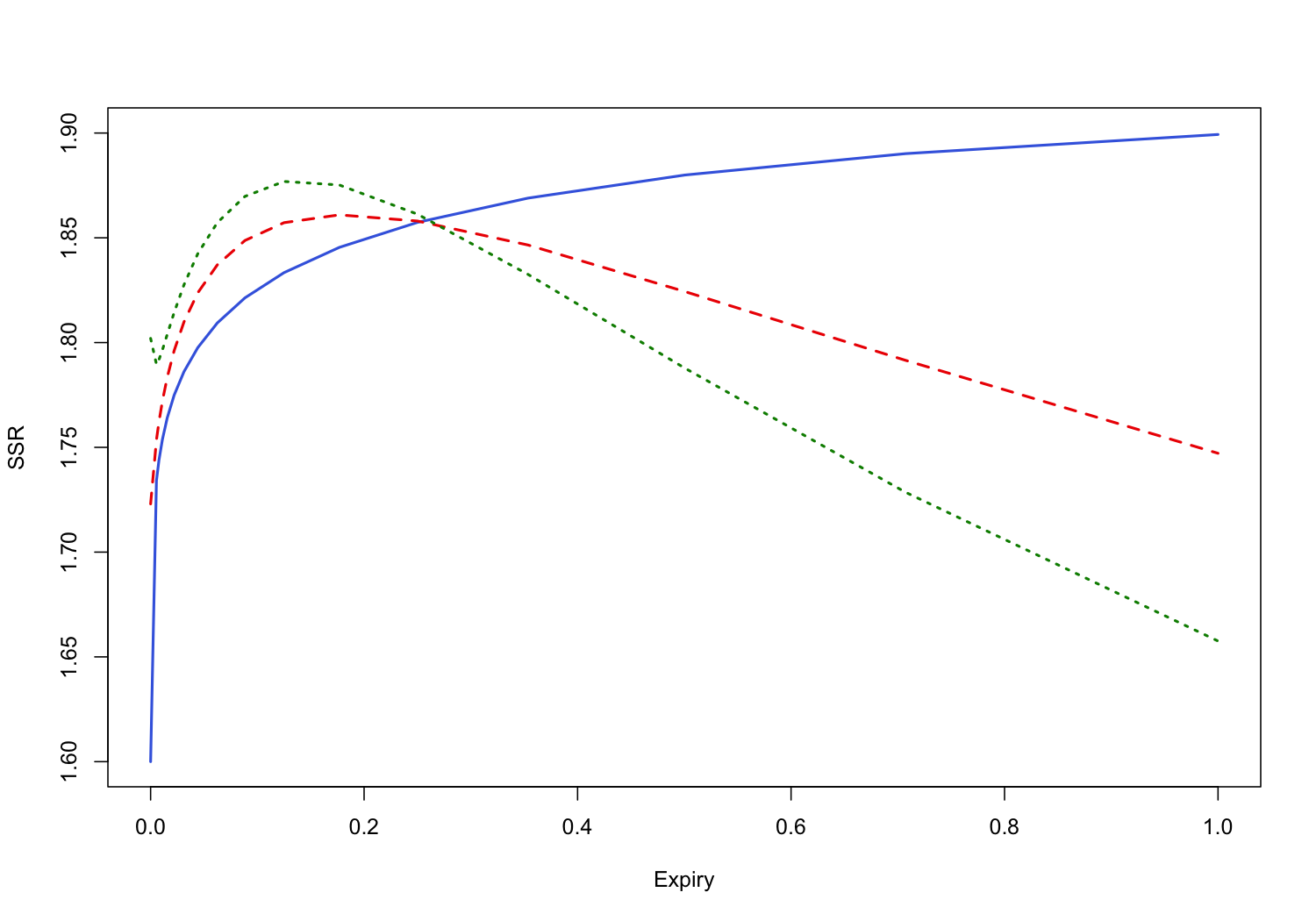}
\caption{The blue, red, and green SSR plots correspond to parameter sets $\Pi_0$, $\Pi_1$, and $\Pi_2$ respectively.  }
\label{fig:SSR3}
\end{center}
\end{figure}

\section{Concluding remarks}\label{sec:conclusion}

To conclude, the key contributions of this paper are twofold.  Firstly, Proposition \ref{prof:SSRandCF}gives a generic and essentially model-free expression for the SSR in terms of the characteristic function.
Secondly, Proposition \ref{prop:SSRaffine} gives an explicit formula for the
SSR in affine forward variance (AFV) models, in terms of the solution of a convolution Riccati integral equation.

By using the forest expansion of the characteristic function, we were able to derive a model-free characterization of the short-time behavior of the SSR in Lemma \ref{lem:leadingorder}.  Corollary \ref{cor:alphaplus1} then confirms various prior formal claims on the limiting behavior of the SSR.

It is important to highlight that observed skew stickiness ratios are typically between 1.0 and 1.3, seemingly inconsistent with our computations in affine models as illustrated in Figure \eqref{fig:SSR3}.  However, we have also established the the SSR is highly dependent on assumed model dynamics and so it is always possible that there exists a model that is consistent with empirical observation. Preliminary simulation results of \cite{bourgey2024skew} suggest that the observed SSR is inconsistent with the most popular rough stochastic volatility models.  This inconsistency remains, in our view, an important outstanding problem. 

In trying to understand where this inconsistency comes from, one naturally asks the question \footnote{A question asked by Eric Reiner in particular.} whether or not this could be related to discreteness in the operational definition of the SSR as opposed to the mathematical definition given by \eqref{eq:SSR} and \eqref{equ:beta}.  We show in Appendix \ref{sec:discreteness} that the discreteness effect cannot explain the large discrepancy between the observed and computed SSRs.

Is is worth noting that within the class of rough volatility models with Hurst parameter $H$, the most reasonable SSR is obtained by taking $H$ very small. Across different rough volatility models (rough Heston, rough Bergomi, quadratic rough Heston) the authors of \cite{bourgey2024skew} have indicated to us that the most realistic term structures of the SSR are generated by the rough Bergomi and quadratic rough Heston with $H$ very small.

%%%%%%%%%%%%%%%%%%%%%%%%%%%%%%%%%%
\bibliographystyle{alpha}
\bibliography{RoughVolatility}

\begin{appendix}

\section{Discreteness effect}\label{sec:discreteness}

In practice, we do not estimate the instantaneous covariance $d \angl{\sigma(T),X}_t$. 
In other words, we have computed what happens to ATM volatility over an infinitesimally short time when the spot moves.  Consider the change in ATM volatility over a short but not infinitesimal period $\Delta$ instead.  The corresponding definition of the regression coefficient would be
\beas
\beta^\Delta_t(T) &= &   \frac{\eetm{\int_t^{t+\Delta}\,d\angl{\sigma(T),X}_s}{\mP}}
{\eetm{\int_t^{t+\Delta}\,d\angl{X}_s}{\mP}},
\eeas
where $\mP$ denotes the physical measure.
The denominator is just integrated variance, estimated in practice by realized variance.  The numerator is estimated by the realized covariance of at-the-money volatility and log-spot.  Both the numerator and the denominator then obviously depend on the choice of measure.

\subsection{Approximation for AFV models}

Consider the following simple change of measure, consistent with, for example, benchmark pricing:
\beas
dZ_t^\mQ = dZ_t^\mP + \sqrt{V_t}\,dt,
\eeas
where $\mQ$ denotes the pricing measure.
For $s \in (0, \Delta)$, and assuming dynamics \eqref{eq:SVmodel}, we make the approximation
\begin{align*}
d\sigma_s(T)^{\mQ}\approx \frac{1}{2\,\sigma\,\tau}\,dM^{\mQ}_s(T) = \frac{1}{2\,\sigma\,\tau}\,\int_s^T\,f_s(\xi)\,\kappa(u-s)\,du\,dW_s^{\mQ},
\end{align*}
where $\sigma$ denotes $\sigma_t(T)^\mQ$ and $M_t(T) = (X \dm X)_t(T) = \int_t^T\,\xi_t(u)\,du$.
Write 
$
dW^\mP_t = \rho\,dZ^\mP_t + \bar \rho\,dZ^{\perp\,\mP}_t
$
and assume the risk associated with $dZ^\mP_t$ is not priced.  Then
$dW_t^{\mQ} = dW_t^{\mP}+\rho\,\sqrt{V_t}\,dt$ so that
\begin{align*}
\eetm{d\angl{\sigma(T),X}_s}{\mP} \approx \frac{\rho}{2\,\sigma\,\tau}\,\int_s^T\,\eetm{\sqrt{V_s}\,f_s(\xi)}{\mP}\,\kappa(u-s)\,du.
\end{align*}
In the AFV case, $f_s(\xi) = \sqrt{V_s}$ and we get
\begin{align*}
\eetm{d\angl{\sigma(T),X}_s}{\mP} \approx \frac{\rho}{2\,\sigma\,\tau}\,\int_s^T\,\eetm{V_s}{\mP}\,\kappa(u-s)\,du.
\end{align*}
Thus, in the AFV case,
\begin{align}
\beta^\Delta_t(T) &=    \frac{\eetm{\int_t^{t+\Delta}\,d\angl{\sigma(T),X}_s}{\mP}}
{\eetm{\int_t^{t+\Delta}\,d\angl{X}_s}{\mP}}\approx 
\frac{\rho}{2\,\sigma\,\tau}\,\frac{\int_t^{t+\Delta}\,ds\,\int_s^T\,\eetm{V_s}{\mP}\,\kappa(u-s)\,du}{\int_t^{t+\Delta}\,ds\,\int_s^T\,\eetm{V_s}{\mP}\,du}.
\label{eq:betaAFVApprox}
\end{align}
In the case of interest, $ \frac{\Delta }{\tau} \ll 1$, it is immediately apparent from expression \eqref{eq:betaAFVApprox} that the effect of the choice of measure and $\Delta$ on $\beta^\Delta_t(T)$ is very small.

\subsubsection{An explicit estimate}
To quantify the size of this effect, consider the special case where $\eetm{V_s}{\mP}=V_t$, a constant.  Then,
\begin{align}
\frac{\beta^\Delta_t(T)}{\beta_t(T)}=\frac{\int_0^{\Delta}\,ds\,\int_s^\tau\,\kappa(u-s)\,du}{\int_0^{\Delta}\,ds\,\int_s^\tau\,du}\,
\frac{\int_0^\tau\,du}
{\int_0^\tau\,\kappa(u)\,du}
\label{eq:betadelover beta}
\end{align}
Noting that
\beas
\int_s^\tau\,\kappa(u-s)\,du = \int_0^{\tau-s}\,\kappa(u)\,du
 %=\int_0^{\tau}\,\kappa(u)\,du- \int_0^{s}\,\kappa(u)\,du 
 \leq \int_0^{\tau}\,\kappa(u)\,du,
\eeas
we obtain
\begin{align*}
\frac{\beta^\Delta_t(T)}{\beta_t(T)}&=\frac{1}{1 -\frac12 \frac{\Delta}{\tau} }\,
\frac{\frac 1 \Delta\,\int_0^{\Delta}\,ds\,\int_s^\tau\,\kappa(u-s)\,du}{\int_0^\tau\,\kappa(u)\,du} 
\leq \frac{1}{1 -\frac12 \frac{\Delta}{\tau} }.
\end{align*}

Thus, in practice, when we estimate the 1-month SSR using daily data, the maximum error relative to the instantaneous estimate \eqref{eq:beta} should be of the order of $2.5\%$, even less for longer expirations.  
If in addition, we were to adjust the denominator of \eqref{eq:SSR} by taking the average skew over the interval $[0,\Delta]$, the discreteness error in $\cR_t(T)$ would be even smaller.

In the realistic case where $\kappa(\tau) = A\,\tau^{-\gamma}$, we may evaluate \eqref{eq:betadelover beta} exactly to give 
\begin{align*}
\frac{\beta^\Delta_t(T)}{\beta_t(T)}&=\frac{1}{1-\frac{\gamma }{2} }\,\frac{1-(1-\epsilon )^{2-\gamma }}{1-(1-\epsilon )^2} = 1+\frac{\gamma  \epsilon }{2}+\cO\left(\epsilon ^2\right),
\end{align*}
where $\epsilon  = \frac \Delta \tau$.  With $\gamma=1/2$, we get half of the upper bound evaluated above.
\end{appendix}

\end{document}